\documentclass[12pt]{article}
\usepackage{amsmath}
\usepackage{amsfonts}
\usepackage{amssymb}
\usepackage{amsthm}
\usepackage{graphicx}
\usepackage{fullpage}
\usepackage{setspace}
\usepackage{verbatim}
\usepackage{natbib}
\usepackage{appendix}
\usepackage{rotating}

\newtheorem{theorem}{Theorem}

\title{A Note on Minimax Testing and Confidence Intervals in Moment Inequality Models}

\author{Timothy B. Armstrong\thanks{email: timothy.armstrong@yale.edu}
\\
Yale University}

\begin{document}

\maketitle

\begin{abstract}
This note uses a simple example to show how moment inequality models used in the empirical economics literature lead to general minimax relative efficiency comparisons.  The main point is that such models involve inference on a low dimensional parameter, which leads naturally to a definition of ``distance'' that, in full generality, would be arbitrary in minimax testing problems.  This definition of distance is justified by the fact that it leads to a duality between minimaxity of confidence intervals and tests, which does not hold for other definitions of distance.  Thus, the use of moment inequalities for inference in a low dimensional parametric model places additional structure on the testing problem, which leads to stronger conclusions regarding minimax relative efficiency than would otherwise be possible.
\end{abstract}

\section{Introduction}\label{introduction_sec}

Recent papers have formulated relative efficiency in moment inequality models in terms of minimax testing and confidence intervals.  Using these definitions of relative efficiency, fairly general statements can be made about how test statistics should be chosen, and about choices of tuning parameters such as weighting functions in the definitions of these test statistics \citep{armstrong_weighted_2014,armstrong_choice_2014}.
This may be surprising to someone familiar with the literature on nonparametric testing
since, in similar nonparametric testing problems, one typically cannot make such general statements about relative efficiency \citep[see, e.g., Chapter 14 of][]{lehmann_testing_2005}.  This is the case even if one restricts attention to using minimaxity as a criterion.  Indeed, the monograph by \citet{ingster_nonparametric_2003} summarizes a large literature that considers different ways of setting up the problem (in terms of the norm and smoothness class) and reaches different conclusions depending on how this is done.
The purpose of this note is to show, in the context of a simple example, how the additional structure of moment inequality models considered in the econometrics literature leads to the relative efficiency results described above.

With this purpose in mind, I consider
notions of minimax testing in moment inequality models that differ in their definition of ``distance from the null.''  In the context of a simple example, I make several points.  First, I argue that a definition of minimax testing related to the excess length of confidence intervals constructed from a test is the most empirically relevant.  I discuss how this definition has a simple, intuitive interpretation.  Second, I point out that this notion of minimaxity corresponds, in this example, to a specific definition used in the nonparametric testing literature, and that using other definitions of minimaxity from that literature would lead to different tests.  Furthermore, I show that confidence intervals constructed from these tests have poor minimax properties.

Thus, for moment inequality models, one can make conclusions about which test should be used that would not be possible without the structure of the problem (in particular, the fact that a test is being inverted to obtain a confidence region for a low dimensional parameter).  To further illustrate this point, I consider a different testing problem involving optimal treatment assignment.  A plausible formulation of the decision problem in the latter application leads to a different definition of minimaxity and a different test than the optimal test for the moment inequality problem, even though the null hypothesis takes the same form.

\subsection{Related Literature}

The main point of this note is essentially an application of the argument
that, in considering estimation or inference for a parameter (defined as a function of a probability distribution that may be defined in a higher dimensional space), it is typically sensible to use the parameter itself in defining a decision theoretic objective function.
The original contribution here is to work out some of the implications of this for moment inequality models when the decision theoretic criterion is minimax; this broader point is not new.
Indeed, it has been an important ingredient in the literature on nonparametric function estimation, and on semiparametric plug-in estimators \citep[for an introduction to these problems, see, e.g.,][]{chen_chapter_2007,ichimura_chapter_2007,tsybakov_introduction_2010}.
It is behind the well known fact that the ``optimal'' bandwidth or number of series terms, etc., differs depending on whether one is interested in a regression discontinuity parameter, weighted average derivative, global estimation of a conditional mean, etc. \citep[see, e.g.,][]{imbens_optimal_2012,powell_optimal_1996,sun_adaptive_2005}.
This point is also behind the work of \citet{dumbgen_optimal_2003}, who uses minimax adaptive tests for global inference on a conditional mean (the fact that a certain notion of minimax nonparametric testing relates to minimax confidence intervals in that context is closely related to the points made here).

This note relates to the recent econometrics literature on moment inequalities and, in particular, papers by
\citet{armstrong_weighted_2014},
\citet{armstrong_choice_2014},
\citet{chernozhukov_testing_2014}
and
\citet{chetverikov_adaptive_2012},
which contain results on minimax testing and confidence intervals.
See also \citet{fang_optimal_2014}, \citet{menzel_consistent_2010} and \citet{song_local_2014} and references therein for results on minimax estimation in related problems.
The broader literature on moment inequalities is too large for a complete review in this section, and I refer the reader to the papers above for references to this literature.
The monograph by \citet{ingster_nonparametric_2003} provides a summary of the literature on minimax testing.
The equivalence results regarding minimax comparisons of tests and confidence intervals in Section \ref{minimax_ci_sec} have, to my knowledge, not been written down explicitly, although they are related to results by \citet{pratt_length_1961} regarding other notions of optimality.

\section{Setup}\label{setup_sec}

To keep things as simple and concrete as possible, I consider a specific model that leads directly to a finite sample normal testing problem with known variances.
The setup below can be considered a simplified version of models considered by \citet{heckman_varieties_1990}, \citet{manski_nonparametric_1990} and
\citet{manski_monotone_2000}.
A researcher is interested in the marginal distribution of wage offers, but only observes wages for people who work, along with an ``instrument'' $X$, that shifts labor force participation but does not affect the distribution of wage offers.  The researcher does not observe the wages of individuals outside of the workforce or even the proportion of such individuals in the population, but makes an assumption of positive selection into the workforce.  Along with the exogeneity assumption for $X$, this can be written as, letting $\theta$ be the marginal expectation of the distribution of wage offers,
\begin{align*}
\theta=E(W^*)=E(W^*|X)\le E(W^*|X, W^* \text{ observed}).
\end{align*}
We observe $X_i$ along with wages $W_i=W_i^*$ for a sample of individuals $i=1,\ldots,n$ in the workforce.  Suppose that $X_i$ takes values in a finite set, which is normalized to $\{1,\ldots,k\}$, and that $\{W_i\}_{i=1}^n$ are independent conditional on $\{X_i\}_{i=1}^n$ (and on being observed) with $W_i|\{X_i\}_{i=1}^n,W_i^*\text{ observed}\sim N(\mu(X_i),\sigma^2)$ for some unknown function $\mu(\cdot)$.

From now on, let us condition on the $X_i$'s and the event that $W_1,\ldots,W_n$ are observed, and use expectation $E(\cdot)$ to denote expectation with respect to the distribution of $W_1,\ldots,W_n$ conditional on $X_1,\ldots,X_n$ and conditional on being observed.  I also condition on $X_1,\ldots,X_n$ in probability and distributional statements from now on, so that, e.g., $W_i\sim N(\mu(X_i),\sigma^2)$ denotes that $W_i$ is $N(\mu(X_i),\sigma^2)$ conditional on $X_i$.  Suppose that $n/k$ is an integer, and that exactly $n/k$ values of $X_i$ take each value $j\in\{1,\ldots,k\}$.  Let $Z_j=\frac{1}{n/k}\sum_{i: X_i=j} W_i$.  To further simplify the problem, assume that $\sigma^2=n/k$, so that $Z_j\sim N(\mu(j),1)$.  This leads to the finite sample moment inequality model
\begin{align}\label{mi_model_eq}
\mu(j)-\theta\ge 0, \, j=1,\ldots,k
\text{ where }
Z\sim N(\mu, I_k), \,
\mu=(\mu(1),\ldots,\mu(k))'.
\end{align}
From now on, I treat $Z$, rather than the $W_i$'s, as the observed data.

The model in (\ref{mi_model_eq}) gives a family of distributions for $Z$ that depends on the unknown parameter $\mu\in\mathbb{R}^k$.  To make this explicit, I index probability statements and expectations with $\mu$, and I use the notation $S=S(Z)\stackrel{\mu}{\sim} \mathcal{L}$ to denote the statement that the statistic $S(Z)$ is distributed with law $\mathcal{L}$ under $\mu$.
The parameter $\mu$ can be thought of as a nuisance parameter, and we are interested in inference on $\theta$.  The identified set for $\theta$ is given by
\begin{align*}
\Theta_0=\Theta_0(\mu)=(-\infty,\min_{1\le j\le k} \mu(j)].
\end{align*}

\section{Confidence Intervals and Minimax Testing}\label{minimax_ci_sec}

Consider the problem of constructing a confidence interval $\mathcal{C}=\mathcal{C}(Z)$ that satisfies the coverage criterion proposed by \citet{imbens_confidence_2004}:
\begin{align}\label{im_cov_eq}
P_\mu(\theta \in \mathcal{C})\ge 1-\alpha
\text{ all }\theta\in \Theta_0(\mu).
\end{align}
Note that, in this setup, if $\mathcal{C}=(-\infty,\hat c]$ for some $\hat c=\hat c(Z)$, which will be the case for the CIs considered in this note, this will be equivalent to the (generally stronger) notion of coverage considered by \citet{chernozhukov_estimation_2007}: $P_\mu(\Theta_0(\mu) \subseteq \mathcal{C})\ge 1-\alpha$.

A confidence region satisfying (\ref{im_cov_eq}) can be obtained by inverting a family of level $\alpha$ tests of
\begin{align*}
H_{0,\theta_0}: %
\theta_0\in\Theta_0(\mu),
\end{align*}
which is equivalent to
\begin{align}\label{ineq_test_eq}
H_{0,\theta_0}: \min_{1\le j\le k} \mu(j)\ge \theta_0.
\end{align}

Consider a family of nonrandomized tests $\phi_{\theta_0}=\phi_{\theta_0}(Z)$ taking the data to a zero-one accept/reject decision for each $\theta_0\in\mathbb{R}$, and a confidence region $\mathcal{C}=\{\theta|\phi_\theta(Z)=0\}$ obtained from inverting these tests.  An obvious question is how to choose between different tests and the associated confidence regions.  Given that considerations such as uniform power comparisons or restrictions to similar-on-the-boundary or unbiased tests are of little use here
(see \citealt{lehmann_testing_1952}, \citealt{hirano_impossibility_2012}, \citealt{andrews_similar---boundary_2012}),
it is appealing to consider minimax comparisons of tests and confidence regions.

One could define both the loss function and the object of interest in several ways.  The object of interest could be the identified set $\Theta_0$, or a particular point $\theta\in\Theta_0$.  The loss function can be defined in terms of Hausdorff distance, the Lebesgue measure of the portion of $\mathcal{C}$ outside of $\Theta_0$ or above a particular point in $\Theta_0$.  While these choices are interesting in general, the simple one-sided nature of this example makes many of them equivalent, which serves the purpose of illustrating ideas in a simple context.  Since the CIs considered in this note will take the form $(-\infty,\hat c(Z)]$, it will be easiest to define the loss function for the confidence region in terms of $\hat c$.
Let $\overline \theta(\mu)=\min_{1\le j\le k}\mu(j)$ so that $\Theta_0(\mu)=(-\infty,\overline\theta(\mu)]$.  Then, the loss function for $\mathcal{C}$ can be defined in terms of $\hat c$ and $\overline \theta$.  Let
\begin{align}\label{loss_func_eq}
\ell(\hat c, \overline \theta)=\tilde \ell((\hat c-\overline \theta)_+)
\end{align}
for a nondecreasing function $\tilde \ell:\mathbb{R}^+\to\mathbb{R}^+$, where $(t)_+=\max\{t,0\}$.  A loss function of the form given above is likely to be a reasonable formulation of the preferences of many researchers and policy makers.  It simply states that smaller values of the upper endpoint, $\hat c$, are preferred so long as coverage is maintained, and treats undercoverage in a neutral manner, since type I error has already been incorporated into the coverage constraint.
The minimax risk of the confidence region $(\infty,\hat c]$ is then
\begin{align}\label{mm_ci_eq}
R(\hat c;\tilde \ell)=\sup_{\mu\in\mathbb{R}^k} E_\mu\ell(\hat c(Z),\overline \theta(\mu))
=\sup_{\mu\in\mathbb{R}^k} E_\mu\tilde \ell((\hat c(Z)-\overline \theta(\mu))_+).
\end{align}
The case of the the zero-one loss function, $\tilde \ell_b(t)=I(t\ge b)$ for some $b>0$, will be of particular interest for its simplicity and its direct relation to minimax hypothesis testing, as will be discussed below.

Let us now consider the formulation of minimaxity for the hypothesis testing problem (\ref{ineq_test_eq}).
In applied work, the goal of performing a test of (\ref{ineq_test_eq}) is often to obtain a confidence region satisfying (\ref{im_cov_eq}).  Thus, to the extent that a loss function of the form (\ref{loss_func_eq}) is reasonable in evaluating these CIs, it is a desirable property for a definition of minimax testing to lead to the same relative efficiency rankings for families of tests that would be obtained by comparing the minimax risk of the associated CIs (\ref{mm_ci_eq}) for some loss function $\tilde \ell$.

Let us consider possible formulations of minimax testing, following Chapter 8 of \citet{lehmann_testing_2005}.  For a given $\theta_0\in\mathbb{R}$, the null region of $H_{0,\theta_0}$ is the set $M_0=M_0(\theta_0)=\{\mu|\mu(j)\ge \theta_0, j=1,\ldots,k\}$.  Size control requires that
\begin{align*}
\sup_{\mu\in M_0} E_\mu\phi_{\theta_0}(Z)\le \alpha.
\end{align*}
Minimax power involves a choice of an alternative set $M_1=M_1(\theta_0)$.  Given this set, the test $\phi_{\theta_0}$ is said to have minimax power at least $\beta$ if
\begin{align}\label{mm_power_eq}
\inf_{\mu\in M_1} E_\mu\phi_{\theta_0}(Z)\ge \beta.
\end{align}
The test is said to have minimax power $\beta$ if the above display holds with equality.
For minimaxity to be interesting, $M_1$ cannot be taken to be the entire alternative set $\mathbb{R}^k\backslash M_0$, since this would lead to trivial minimax power ($\beta=\alpha$).  Thus, in full generality, minimax testing involves a degree of arbitrariness in specifying $M_1$, which has been a criticism against its use.

One of the main points of this note is to argue that, for the problem considered here, this decision is not arbitrary, and a particular class of alternatives $M_1$ should be used.  This is because of its relation with the minimax risk (\ref{mm_ci_eq}) for the associated CI, which, as argued above, is a desirable property.  Given $\theta_0$ and a positive scalar $b$, define the alternative
\begin{align}\label{mm_alt_eq}
M_1^*(\theta_0,b)=\{\mu|\theta\le \theta_0-b \text{ all } \theta\in\Theta_0(\mu)\}
=\{\mu|\overline \theta(\mu)\le \theta_0-b\}.
\end{align}
Here, $b$ is a constant that defines distance to the null, which can be calibrated so that the minimax power $\beta$ of the test is above a certain level.

It will be shown below that relative efficiency comparisons for minimax power based on (\ref{mm_alt_eq}) have a duality with relative efficiency comparisons for the corresponding CIs based on (\ref{mm_ci_eq}).  Before doing so, I make two additional points supporting the usefulness of minimaxity with the alternative $M_1^*(\theta_0,b)$.  First, note that a test $\phi_{\theta_0}$ with level $\alpha$ for $H_{0,\theta_0}$ and minimax power $\beta$ for $M_1^*(\theta_0,b)$ controls both type I error under the null $\theta_0\in\Theta_0(\mu)$ and type II error uniformly over the set of data generating processes (dgps), indexed by $\mu$, such that $\theta_0$ exceeds any possible value of $\theta$ consistent with the data generating process by at least $b$.  This has a simple interpretation that can be explained to an applied researcher
(e.g. ``you wrote down a model that would give some upper bound, $\overline \theta$, for the mean offer wage if we had the entire population; given your sample size and the test that you are using, you will be able to determine that $\theta_0$ is greater than this upper bound at least $75\%$ if the time, so long as $\theta_0$ is greater than the upper bound by at least $\$10,000$ per year'').

Second, the definitions of null and alternative can be reversed, yielding a one-sided test for $\overline \theta(\mu)$ in the other direction, and a confidence region giving a lower bound for $\overline \theta(\mu)$.  This can be used to quantify how much of the length of the one-sided interval $(-\infty,\hat c]$ is due to statistical uncertainty, and how much is due to the population upper bound $\overline \theta(\mu)$ being large.  It provides an answer to questions such as: ``should I get a larger sample size, or should I search for a different empirical strategy, perhaps with stronger assumptions?''

\subsection{Duality Between Minimaxity for CIs and Tests}

I now state two theorems giving a duality between the definitions of minimax testing and confidence intervals defined above.  I begin with a result for zero-one loss functions.

\begin{theorem}\label{duality_thm}
Let $\phi_{\theta_0}$ be a class of nonrandomized level $\alpha$ tests for the family (\ref{ineq_test_eq}), with associated confidence region $(-\infty,\hat c]$.  Let $\beta_{\theta_0}(b)$ be the minimax power of the test of $H_{0,\theta_0}$ for the alternative $M_1^*(\theta_0,b)$.
Then, for the zero-one loss function $\tilde \ell_b(t)=I(t\ge b)$,
\begin{align*}
\inf_{\theta_0\in\mathbb{R}} \beta_{\theta_0}(b)
=1-R(\hat c,\tilde \ell_b).
\end{align*}

\end{theorem}
\begin{proof}
We have
\begin{align*}
&\beta_{\theta_0}(b)
=\inf_{\mu\in M_1^*(\theta_0,b)} E_\mu \phi_{\theta_0}
=\inf_{\mu \text{ s.t. } \overline \theta(\mu)\le \theta_0-b} E_\mu \phi_{\theta_0}
=\inf_{\mu \text{ s.t. } \overline \theta(\mu)\le \theta_0-b} P_\mu (\theta_0\notin (-\infty,\hat c])  \\
&=\inf_{\mu \text{ s.t. } \overline \theta(\mu)\le \theta_0-b} P_\mu (\theta_0>\hat c)
=1-\sup_{\mu \text{ s.t. } \overline \theta(\mu)\le \theta_0-b} P_\mu (\hat c-\theta_0\ge 0)  \\
&=1-\sup_{\mu \text{ s.t. } \overline \theta(\mu)\le \theta_0-b} E_\mu I(\hat c-\theta_0+b\ge b)
\end{align*}
Taking the infimum of both sides over $\theta_0$ gives
\begin{align*}
&\inf_{\theta_0\in\mathbb{R}} \beta_{\theta_0}(b)
=1-\sup_{\theta_0\in\mathbb{R}}\sup_{\mu \text{ s.t. } \overline \theta(\mu)\le \theta_0-b} E_\mu I(\hat c-\theta_0+b \ge b)
=1-\sup_{\mu\in\mathbb{R}^k} E_\mu I(\hat c-\overline \theta(\mu)\ge b)  \\
&=1-R(\hat c,\tilde \ell_b),
\end{align*}
where the second equality follows by switching the order of the suprema and noting that, for a given $\mu$,
$\sup_{\theta_0\text{ s.t. } \overline \theta(\mu)\le \theta_0-b} E_\mu I(\hat c-\theta_0+b\ge b)=E_\mu I(\hat c-\overline \theta(\mu)\ge b)$.

\end{proof}

While the zero-one loss functions $\tilde \ell_b(t)=I(t\ge b)$ are intuitively appealing, one may wish to consider a more general nondecreasing function $\tilde \ell(t)$.  This can be related to the zero-one loss functions (and therefore minimax power as well), so long as
the same distribution $\mu$ is simultaneously least favorable for each $\tilde \ell_b$.

\begin{theorem}\label{loss_func_thm}
For any increasing function $\tilde \ell:\mathbb{R}^+\to\mathbb{R}^+$, there exists a measure $\nu_{\tilde \ell}$ on $\mathbb{R}^+$ such that the following holds.
For any CI $(-\infty,\hat c]$ such that there exists a parameter value $\mu^*$ that is simultaneously least favorable for all zero-one loss functions $\tilde \ell_b$,
\begin{align*}
R(\hat c,\tilde \ell)=\int R(\hat c,\tilde \ell_b)\, d\nu_{\tilde \ell}(b)
\end{align*}
\end{theorem}
\begin{proof}
Let $\nu_{\tilde \ell}$ be such that $\tilde\ell(t)=\int \tilde \ell_b(t)\, d\nu_{\tilde \ell}(b)$.
Then
\begin{align*}
&R(\hat c,\tilde \ell)
\ge E_{\mu^*} \tilde \ell((\hat c-\overline \theta(\mu^*))_+)
=E_{\mu^*} \int \tilde \ell_b((\hat c-\overline \theta(\mu^*))_+)\, d\nu_{\tilde \ell}(b)
= \int E_{\mu^*}\tilde \ell_b((\hat c-\overline \theta(\mu^*))_+)\, d\nu_{\tilde \ell}(b)  \\
&= \int R(\hat c,\tilde \ell_b)\, d\nu_{\tilde \ell}(b)
\end{align*}
using Fubini's theorem and the fact that
$R(\hat c,\tilde\ell_b)=E_{\mu^*}\tilde \ell_b((\hat c-\overline \theta(\mu^*))_+)$
by simultaneous least favorability of $\mu^*$ for all $b$.  Similarly, if the inequality in the above display were strict, we would have, for some other $\mu$,
\begin{align*}
\int E_{\mu^*}\tilde \ell_b((\hat c-\overline \theta(\mu^*))_+)\, d\nu_{\tilde \ell}(b) <E_{\mu} \tilde \ell((\hat c-\overline \theta(\mu))_+)
=\int E_{\mu}\tilde \ell_b((\hat c-\overline \theta(\mu))_+)\, d\nu_{\tilde \ell}(b),
\end{align*}
which would imply $E_{\mu}\tilde \ell_b((\hat c-\overline \theta(\mu))_+)>E_{\mu^*}\tilde \ell_b((\hat c-\overline \theta(\mu^*))_+)$ for some $b$, thereby violating simultaneous least favorability of $\mu^*$.

\end{proof}

While the simultaneous least favorability condition used above can be strong in some applications, it will hold in the simple cases I consider here.

\subsection{Minimax Efficiency of Some Popular Tests}

Using the definition of minimaxity developed above, I now compare the relative efficiency of some commonly used tests, and give upper bounds for the minimax power of any test.  I show that the upper bounds satisfy a certain asymptotic sharpness property in the large $k$ case.

Consider the following class of test statistics, based on the one-sided $L^p$ norm for $p\ge 1$:
\begin{align*}
S_p(Z,\theta_0)=\left(\sum_{j=1}^k (\theta_0-Z_j)_+^p\right)^{1/p}
\text{ and } S_\infty(Z,\theta_0)=\max_{1\le j\le k} (\theta_0-Z_j)_+
\end{align*}
and tests based on the least favorable level $\alpha$ critical value
\begin{align*}
c_{p,\alpha}=\sup_{\mu\in\mathbb{R},\theta_0\le \min_{1\le j\le k} \mu(j)}
q_{\mu,1-\alpha}(S_p(Z,\theta_0))
=q_{(0,\ldots,0),1-\alpha}(S_p(Z,0)),
\end{align*}
where $q_{\mu,\tau}$ denotes the $\tau$th quantile under $\mu$, and the equality in the above display follows since increasing $\theta_0$ and decreasing elements of $\mu$ stochastically increases $S_p(Z,\theta_0)$, and the distribution of the test statistic is invariant to adding the same quantity to $\theta_0$ and all components of $\mu$.
The test $\phi_{\theta_0,p}$ is defined as the nonrandomized test that rejects when $S_p(Z,\theta_0)>c_{p,\alpha}$.
Note that this test does not incorporate moment selection (see, e.g., \citealt{hansen_test_2005}, \citealt{andrews_inference_2010} for definitions of moment selection procedures), although some of the analysis below allows for such procedures (note, in particular, that Theorem \ref{upper_bound_thm} gives an asymptotic optimality result among all tests, including those that incorporate moment selection).
Let $\beta_{\theta_0,p}(b)$ be the minimax power for the alternative (\ref{mm_alt_eq}).
Define $\Phi$ to be the standard normal cdf.

\begin{theorem}\label{lp_minimax_thm}
The minimax power of $\phi_{\theta_0,p}$ is given by
\begin{align*}
\beta_{\theta_0,p}(b)
=P_{(-b,\infty,\ldots,\infty)}(S_p(Z,0)>c_{\alpha,p})
=1-\Phi(c_{\alpha,p}-b)
\end{align*}
for $\alpha\le 1/2$.
It is strictly increasing in $p$ for $1\le p\le \infty$.
\end{theorem}
\begin{proof}
The first claim follows by symmetry and by noting that power is decreasing in each $\mu(j)$, since the distribution of the test statistic is stochastically decreasing in each $\mu(j)$.
Note that $\alpha\le 1/2$ implies $c_{\alpha,p}\ge 0$, which gives the second equality in the display.
The second claim follows by noting that $c_{\alpha,p}$ is strictly decreasing in $p$, which follows since the positive orthant of $L^p$ balls are (strictly) contained in the positive orthant of $L^q$ balls for $p< q$ (when the radius is the same and both are centered at the origin).
\end{proof}

By Theorem \ref{duality_thm}, it follows that the $L^\infty$ test leads to the best minimax confidence region among $L^p$ statistics for any zero-one loss function.  In fact, Theorem \ref{lp_minimax_thm} shows that the same distribution $(-b,\infty,\ldots,\infty)$ is least favorable distribution for zero-one loss $\tilde \ell_b$ for any $b$.  Thus, Theorem \ref{loss_func_thm} applies, and the $L^\infty$ CI is optimal for any increasing loss function.

I now state an upper bound on minimax power for any level $\alpha$ test, which is essentially a restatement of results in \citet{dumbgen_multiscale_2001} and \citet{chernozhukov_testing_2014} \citep[the upper bound apparently goes back at least to][]{ingster_asymptotically_1993}.
While the results in those papers are for minimax testing in the $L^\infty$ norm, the results translate immediately to our setting. This is because, for the example considered here, minimax one sided inference on $\theta$ is equivalent to one-sided $L^\infty$ minimaxity in a nonparametric testing problem.  I discuss this equivalence further in Section \ref{norm_comp_sec} below.

\begin{theorem}\label{upper_bound_thm}
For any level $\alpha$ test $\phi_{\theta_0}$ of (\ref{ineq_test_eq}), the minimax power $\beta_{\theta_0}(b)$ is bounded by
\begin{align*}
\overline \beta(b;k)= P_{(-b,0,\ldots,0)}\left(\sum_{j=1}^k \exp(-Z_jb)> \tilde c_{\alpha}\right)
\end{align*}
where
\begin{align*}
\tilde c_{\alpha}=q_{(0,0,\ldots,0),1-\alpha}\left(\sum_{j=1}^k \exp(-Z_jb)\right).
\end{align*}
Furthermore, as $k\to \infty$,
\begin{align*}
\overline \beta(\sqrt{(2-\varepsilon)\log k};k)\to \alpha
\end{align*}
and
\begin{align*}
\beta_{\theta_0,\infty}(\sqrt{(2+\varepsilon)\log k};k)\to 1
\end{align*}
for any $\varepsilon>0$,
where $\beta_{\theta_0,\infty}(b;k)$ denotes the minimax power of the $L^\infty$ test for $M_1^*(b,\theta_0)$ for the given value of $k$.
\end{theorem}
\begin{proof}
The minimax power cannot be greater than the power of the most powerful test of $(\theta_0,\ldots,\theta_0)$ against the alternative that places weight $1/k$ on each $(\theta_0,\ldots,\theta_0,\theta_0-b,\theta_0,\ldots,\theta_0)$, where the position of $\theta_0-b$ ranges from $1$ to $k$.  By symmetry, the minimax power for this subproblem is equal to the minimax power in the same problem when $\theta_0=0$.  The first result follows by calculating the Neyman-Pearson test for this problem.

The second claim follows by Lemma 6.2 in
\citet{dumbgen_multiscale_2001} \citep[see Section 5 of][]{chernozhukov_testing_2014}.
The final result follows since $c_{\alpha,\infty}/\sqrt{2\log k}\to 1$ as $k\to \infty$ (note that the dependence of $c_{\alpha,\infty}$ on $k$ is supressed in the notation).
\end{proof}

Theorem \ref{upper_bound_thm} shows that the $L^\infty$ test is approximately optimal in a certain sense for large $k$.  As $k\to \infty$, no test can have nontrivial minimax power against all alternatives that deviate from the null by $\sqrt{(2-\varepsilon) \log k}$, and the $L^\infty$ test has minimax power approaching one as $k\to \infty$ for alternatives that deviate from the null by $\sqrt{(2+\varepsilon)\log k}$.  Equivalently, any CI must increase at least proportionally to $\sqrt{(2-\varepsilon)\log k}$ for some sequence of distributions, and the $L^\infty$ CI increases proportionally to $\sqrt{(2+\varepsilon)\log k}$ or less for all sequences of distributions.
Note also that this result shows that the $L^\infty$ test is close to optimal even without the moment selection procedures mentioned above.

\section{Comparison to Other Notions of Minimax Testing}\label{norm_comp_sec}

The nonparametric statistics literature has considered numerous definitions of minimax tests in problems similar to the one considered here (see \citealt{ingster_nonparametric_2003} for an overview of this literature).  To put these into the context of the present setup, consider testing $H_{0,\theta_0}$ when $\theta_0=0$.  Then, the null hypothesis takes the form $\mu(j)\ge 0$ all $1\le j\le k$.  Suppose that we are interested in testing this null hypothesis in an abstract sense, without the model of Section \ref{setup_sec} to guide our choice of alternatives.

As discussed in Section \ref{minimax_ci_sec}, one needs to separate the alternative hypothesis from the null in order for minimaxity to be interesting.  Without the structure of a low dimensional parametric model like the one defined in Section \ref{setup_sec}, one can imagine many ways of doing this.  A popular choice in the minimax testing literature is to use the $L^p$ norm or, in our case, the one-sided $L^p$ norm\footnote{The formulation given here is a natural extension of the setup for testing the simple null $H_0:\mu=0$ using minimaxity with respect to the $L^p$ norm considered in the literature described in \citet{ingster_nonparametric_2003}.  However, I am not aware of papers in this literature considering the one-sided case considered here for $p<\infty$.  One sided minimax optimality in the $L^\infty$ norm, has been considered by \citet{dumbgen_multiscale_2001}, \citet{chetverikov_adaptive_2012} and
\citet{chernozhukov_testing_2014}.}
\begin{align*}
&\|x\|_{-,p}=\left(\sum_{j=1}^k (-x_k)_+^p\right)^{1/p}
\text{ for } 1\le p<\infty,  \\
&\|x\|_{-,\infty}=\max_{1\le j\le k} (-x_k)_+,
\end{align*}
and to use this as a notion of distance to separate the null and alternative.  Define
\begin{align*}
M_{1,p}(b)=\{\mu| \|\mu\|_{-,p}\ge b\}.
\end{align*}
The minimax power of a test $\phi$ is then given by
\begin{align*}
\beta_{p}(b)=\inf_{\mu\in M_{1,p}(b)}E_\mu \phi.
\end{align*}

I now make two points regarding the relation between these notions of minimaxity and the notion of minimaxity based on distance to the identified set, which was developed in Section \ref{minimax_ci_sec} and shown to correspond to a certain form of minimaxity of confidence intervals.  First, note that the definition in the above display for $p=\infty$ is the same as the one in Section \ref{minimax_ci_sec}: $M_{1,\infty}(b)=M_{1}^*(0,b)$.  Thus, one can interpret the analysis in Section \ref{minimax_ci_sec} as using the structure of the model to choose a one-sided norm in an abstract definition of minimaxity.
The close connection between $L^\infty$ minimaxity and confidence intervals for $\theta$ in this problem should not be surprising, given that
\citet{armstrong_weighted_2014} and \citet{chetverikov_adaptive_2012}
arrived at essentially the same prescription for which tests should be used in a more general version of the problem considered in this paper, with
\citet{armstrong_weighted_2014} considering minimaxity of confidence intervals
and \citet{chetverikov_adaptive_2012} considering $L^\infty$ minimaxity for the conditional mean.

Second, minimaxity with other definitions of distance would, in general, lead to different prescriptions for the optimal test.  While minimax power comparisons and characterizations of minimax optimal or near minimax optimal tests do not appear to be available in the literature for the one-sided $L^p$ norm for $p<\infty$ (in contrast to numerous results considering two-sided testing in the $L^p$ norm), some limited Monte Carlo experiments indicate that some of the other tests considered in Section \ref{minimax_ci_sec} have better minimax power than the max test under $L^p$ norms with $p<\infty$.

\section{A Testing Problem Leading to $L^1$ Minimaxity}

In the previous sections, I have argued that one should use $\theta$ to define minimaxity in cases where tests are being inverted to obtain a confidence region for this parameter.  It was shown that, for the problem considered in this paper, this notion of minimaxity happened to coincide with a definition of minimaxity involving the one-sided $L^\infty$ norm on the conditional mean.  It was also pointed out that the results of \citet{armstrong_weighted_2014} and \citet{chetverikov_adaptive_2012} suggest that there is a close connection between these definitions of minimaxity in problems of this type encountered in empirical economics more generally.

It should be emphasized that the connection to $L^\infty$ minimaxity is a feature of this and other parametric moment inequality models considered in empirical economics (at least those satisfying conditions given in \citealt{armstrong_weighted_2014}), rather than a general justification for considering the $L^\infty$ norm when defining minimaxity in one-sided testing problems.  To illustrate this point, I now describe a testing problem where a plausible specification of a researcher's utility function leads to the one-sided $L^1$ (rather than $L^\infty$) norm in the definition of minimaxity, even though
the null hypothesis is formally the same (nonpositivity of a mean vector) up to a sign normalization.  The motivation for this notion of minimaxity is, perhaps, less strong in this problem, and the description that follows should not be taken as an argument for its use.  Rather, the point is simply to show that reasonable formulations of minimaxity do not always lead to some variant of the $L^\infty$ norm, even though this appears to be the case in many problems arising in inference on set identified parameters in empirical economics.

Consider an optimal treatment assignment problem following \citet{manski_statistical_2004}, with a sample stratified by a finitely discrete variable taking on values normalized to $\{1,\ldots,k\}$.  Rather than optimal treatment assignment, I consider a related hypothesis testing problem.  The researcher observes outcome variables $Y=Y_1D+Y_0(1-D)$ along with the variable $X\in\{1,\ldots,k\}$, where $D$ is an indicator for treatment, and an unconfoundedness assumption is assumed to hold: $E(Y_j|X,D)=E(Y_j|X)$ for $j\in\{0,1\}$.  The goal is to find a treatment rule $r:\{1,\ldots,k\}\to \{0,1\}$ maximizing
\begin{align}\label{treatment_obj}
E(u(Y_{r(X)}))  
\end{align}
where $u(y)$ is the Bernoulli utility that the social planner assigns to the outcome $y$ for a given individual.
In practice, the data generating process is unknown,
and the treatment rule is based on a random sample $\{(X_i,Y_i,D_i)\}_{i=1}^n$,
and the risk of the treatment rule $\hat r$ is obtained by plugging in $\hat r$ to (\ref{treatment_obj}) and integrating over the data generating process for the sample that leads to $\hat r$.

Suppose that, rather than (or in addition to) recommending a treatment rule, a researcher or policy maker is interested in testing the null hypothesis that no individuals should be treated.  That is, the null hypothesis is that the treatment rule $r$ that would maximize (\ref{treatment_obj}) given full knowledge of the data generating process sets $r(j)=0$ for all $j$.  For simplicity, let us assume that $u(t)=t$ (i.e. $Y$ is already measured in units of Bernoulli utility), and suppose that
we observe a sample $\{(X_i,Y_i,D_i)\}_{i=1}^n$, where
$n$ is a multiple of $2k$, and the sample has $n/(2k)$ observations with $X_i=j$ and $D_i=\ell$ for each $j\in\{1,\ldots,k\}$ and $\ell\in\{0,1\}$.  We will treat the $X_i$'s in the sample as nonrandom, so that we require size control and evaluate power conditional on the $X_i$'s.  We assume that $X$ is distributed uniformly on $\{1,\ldots,k\}$ in the population, so that the expectation in (\ref{treatment_obj}) is evaluated with respect to a uniform distribution on $X$.

Let $\tau(x)=E(Y_1|X=x)-E(Y_0|X=x)$ for $x\in\{1,\ldots,k\}$.  Suppose that $Y_i$ is normal with (known) variance $n/(4k)$, so that
\begin{align*}
Z_j=\frac{1}{n/(2k)}\sum_{X_i=j,D_i=1} Y_i-\frac{1}{n/(2k)}\sum_{X_i=j,D_i=0}Y_i
\sim %
N(\tau(j),1)
\end{align*}
for $j=1,\ldots,k$.  With this notation and set of assumptions, we have, for a treatment rule $r$,
\begin{align*}
E(u(Y_{r(X)}))-E(u(Y_{0}))
=\frac{1}{k}\sum_{j=1}^k \tau(j)r(j).
\end{align*}
Given knowledge of the data generating process, the treatment rule that would maximize (\ref{treatment_obj}) would simply set $r(j)=1$ for $\tau(j)>0$ and $0$ otherwise.  Thus, letting $r^*$ be this treatment rule, the gain in welfare from using $r^*$ relative to a rule that assigns nontreatment to the entire population is
\begin{align}\label{welfare_gain_eq}
w^*=w^*(\tau)=E(u(Y_{r^*(X)}))-E(u(Y_{0}))
=\frac{1}{k}\sum_{j=1}^k (\tau(j))_+
=\frac{1}{k}\|\tau\|_{+,1}
\end{align}
where $\tau$ is the vector with $j$th component $\tau(j)$ and $\|x \|_{+,1}=\sum_{j=1}^k (x_k)_+$ is the positive $L^1$ norm, analogous to the negative $L^1$ norm defined in Section \ref{norm_comp_sec}.

Thus, the null hypothesis that $r^*(j)=0$ all $j$ can be written as
\begin{align}\label{te_null_eq}
H_0: \tau(j)\le 0 \text{ all } j.
\end{align}
As discussed in Section \ref{minimax_ci_sec}, one must separate the alternative from the null in order for minimax testing to be interesting.  Consider the alternative set defined using the welfare gain $w^*$ from population optimal treatment assignment:
\begin{align*}
M_1(b)=\{\tau| w^*(\tau)\ge b\}.
\end{align*}
For a test $\phi$ with level $\alpha$ for the null (\ref{te_null_eq}), the minimax power is then
\begin{align*}
\beta^*(b,\phi)=\inf_{\tau\in M_1(b)}E_{\tau} \phi.
\end{align*}
Note that, in contrast to the definition of minimaxity that came out of considering confidence intervals in the  moment inequality model defined in Section \ref{minimax_ci_sec}, this definition of minimax coincides with the one based on the one-sided $L^1$ norm, rather than the one-sided $L^\infty$ norm (with the obvious change of signs).

This example is somewhat contrived, and I do not wish to argue for the adoption of $\beta^*(b,\phi)$ for hypothesis testing problems related to treatment assignment.  Rather, I would like to argue that (1) $\beta^*(b,\phi)$ has a simple economic interpretation that can be useful in understanding the properties of a test $\phi$, (2) relative power comparisons based on $\beta^*(b,\phi)$ arise from a reasonable specification of a researcher's utility and (3) a definition of minimax in terms of the one-sided $L^\infty$ norm would lead to neither of these properties.

Regarding point (1), suppose that a researcher is interested in treatment effect heterogeneity and wants to explain the results of a test $\phi$ to a social planner, who is contemplating statistical treatment rules and has the preferences formulated above.  The social planner wants to know how good the test is at detecting data generating processes for which treatment of some individuals is desirable.  The researcher can explain as follows. ``The test $\phi$ will reject only $100\cdot \alpha$\% of the time if there is no gain to treatment.  If it is possible to design a treatment rule that improves on nontreatment by at least $b$ expected utils after collecting a very large amount of data, the test will reject with probability at least $100\cdot \beta^*(b,\phi)$\%.''

Regarding point (2), consider the following decision problem.  The same researcher and social planner described above are deciding whether to pursue a social program based on the data described above.  If they decide to pursue the project, another team of researchers will collect a very large amount of data, and will implement the (population) optimal treatment rule.  The only constraint is that the other team of researchers has a limit on the proportion of ultimately fruitless projects that they will pursue, say, $100\cdot \alpha$\%.
If the potential welfare gains from a project are small, say $w^*(\tau)<b$, and the project is not pursued, it will go unnoticed.  However, if there are large welfare gains, say, $w^*(\tau)\ge b$, a different team of researchers from a rival political party will discover this and use it to damage the social planner's reputation.  With this formulation, the researcher and social planner using a test that maximizes $\beta^*(b,\phi)$ to decide whether to pursue a social program is a minimax decision in the sense of minimizing the worst case expected loss.
Of course, certain assumptions in this formulation are unrealistic (e.g. if there are winners and losers in the social program, the proportion of the population with $Y_1>Y_0$ would be more relevant for the political reputation of the social planner than average welfare).  The point is simply that there is some decision problem related to the optimal treatment assignment problem that leads to $\beta^*(\beta,\phi)$ as a criterion for choosing between tests.

Regarding point (3), since there is no direct relation between the one-sided $L^\infty$ norm and expected welfare, defining minimaxity in this way would not lead to the same interpretations for minimax tests.  Suppose that a researcher wanted to explain to a social planner the power properties of a test in terms of minimax one-sided $L^\infty$ power.  The researcher would translate $L^\infty$ minimaxity to the social planner by making a statement along the lines of: ``as long as there exists a $j$ such that individuals with $X=j$ benefit from the program by at least $b$, the test will reject a certain percentage of the time.''  The social planner would likely object, saying: ``but I explained to you that my objective function was average welfare, and you're describing the welfare of those with the value of $X$ who benefit the most.''  Similarly, while one could formulate a decision problem for which $L^\infty$ minimaxity makes sense, it would not be related directly to the social planner's objective function involving expected welfare.

Thus, the treatment assignment problem described above leads to a different one-sided norm in a reasonable definition of minimax testing.  While the motivation for minimax testing in this example is somewhat contrived (arguably much more so than in the example in Section \ref{minimax_ci_sec}, where minimaxity was related directly to a confidence interval for a parameter of interest), these examples illustrate the point that an appropriate definition of minimax testing depends on the structure of the economic problem.

\section{Conclusion}

This note has used a simple example to show how the structure of moment inequality models used in economics leads to relative efficiency results.  A low dimensional parametric model defines a natural distance for minimax testing, which has a duality with minimax risk for the corresponding confidence intervals.  The low dimensional parametric model can be interpreted as providing a justification for using a particular notion of ``distance'' in defining minimax testing, which would be arbitrary in a more general setup.

\bibliography{library}

\end{document}